\documentclass[smallextended]{svjour3}       
\smartqed  

\usepackage{amssymb,amsmath}
\usepackage{graphicx}
\usepackage{hyperref}
\usepackage{enumerate, multirow}

\newcommand{\seq}[1]{\mathbf{#1}}
\newcommand{\set}[1]{\mathcal{#1}}
\newcommand{\C}{\mathcal{C}}
\newcommand{\I}{{\mathcal{I}}}

\begin{document}
\title{New CRT sequence sets for a collision channel without feedback}

\author{Yijin Zhang \and Yuan-Hsun Lo \and \\ Kenneth W. Shum \and Wing Shing Wong}


\institute{Y. Zhang \at
    School of Electronic and Optical Engineering, Nanjing University of Science and Technology, Nanjing, China \\
    National Mobile Communications Research Laboratory, Southeast University, Nanjing, China \\
    \email{yijin.zhang@gmail.com}
\and
    Y.-H. Lo \at
    School of Mathematical Sciences, Xiamen University, Xiamen, China\\
    \email{yhlo0830@gmail.com}
\and
    K. W. Shum \at
    Institute of Network Coding, The Chinese University of Hong Kong, Hong Kong\\
    \email{wkshum@inc.cuhk.edu.hk}
\and
    W. S. Wong \at
    Department of Information Engineering, The Chinese University of Hong Kong, Hong Kong\\
    \email{wswong@ie.cuhk.edu.hk}
\and The material in this paper was presented in part at Globecom 2014 Workshop on ULTRA, Austin, USA, Dec. 2014.
}

\date{Received: date / Accepted: date} 

\renewcommand{\labelenumi}{(\roman{enumi})}

\maketitle

\begin{abstract}
Protocol sequences are binary and periodic sequences used for deterministic multiple access in a collision channel without feedback.
In this paper, we focus on user-irrepressible (UI) protocol sequences that can guarantee a positive individual throughput per sequence period with probability one for a slot-synchronous channel, regardless of the delay offsets among the users.
As the sequence period has a fundamental impact on the worst-case channel access delay, a common objective of designing UI sequences is to make the sequence period as short as possible.
Consider a communication channel that is shared by $M$ active users, and assume that each protocol sequence has a constant Hamming weight $w$.
To attain a better delay performance than previously known UI sequences, this paper presents a CRTm construction of UI sequences with $w=M+1$, which is a variation of the previously known CRT construction.
For all non-prime $M\geq 8$, our construction produces the shortest known sequence period and the shortest known worst-case delay of UI sequences.
Numerical results show that the new construction enjoys a better average delay performance than the optimal random access scheme and other constructions with the same sequence period, in a variety of traffic conditions.
In addition, we derive an asymptotic lower bound on the minimum sequence period for $w=M+1$ if the sequence structure satisfies some technical conditions, called equi-difference, and prove the tightness of this lower bound by using the CRTm construction.



\end{abstract}


\section{Introduction}\label{sec:intro}

\subsection{Background}\label{sec:intro_background}
Protocol sequences are periodic deterministic binary sequences that are used to provide reliable medium access control (MAC) protocol for a collision channel without feedback~\cite{Massey85}.
Compared with time division multiple access (TDMA), ALOHA and carrier sense multiple access (CSMA), a protocol sequence-based scheme does not require stringent time synchronization, channel monitoring, backoff algorithm or packet retransmission.
Such simplicity is particularly desirable in wireless sensor networks (WSNs) and vehicular ad hoc networks (VANETs), where well-coordinated transmission and time synchronization may be difficult to achieve due to user mobility, time-varying propagation delays or energy constraints.
The natural interest of protocol sequence based schemes, the guaranteed performance metrics, such as worst-case delay or minimum throughput, have been commonly considered in previous studies on protocol sequences \cite{Wong07,CWS08,SCSW,VANET13,VANET14,Wong14,MS16,XCW16,CYCLWK16,ZLWS16,SCCW16}. 
In addition, some further performance metrics, such as average group/individual delay or average throughput, have been investigated in \cite{SCSW,ZLWS16,SCCW16}, and related approaches for sequence allocation can be found in \cite{VANET14,Wong14,MS16,XCW16}. 

In this paper, our focus is on {\em user-irrepressible} (UI) protocol sequences that can guarantee a positive individual throughput per sequence period with probability one for a slot-synchronous channel. 
UI property is a fundamental requirement in a protocol sequence-based scheme for delay-constrained services with small amounts of user data.
The design goal of UI sequences is to minimize the sequence period, which indicates how long the receiver has to wait between two successfully transmitted packets in the worst-case, when the number of active users are given.
Some constructions of UI sequences can be found in~\cite{Massey85,Wong07,CWS08,SCSW,CYCLWK16,Nguyen92,GyorfiVajda93,BE95,YK95,YK02,SWSC,CRT,UIS09}.

We remark that deterministic MAC protocols can also be referred to in the literature as conflict-avoiding codes (CACs)\cite{LT05,CAC32010,SWC10,LFL15,LMJ16}, optical orthogonal codes (OOCs)\cite{CSW89,WJ10,FWWZ16}, or topological transparent scheduling \cite{CF94,CCS06,CCS06-2,LLLZ14} with different design goals.
In particular, UI sequences aim to minimize the sequence period by assuming all users are active, whereas CACs aim to maximize the number of potential users when the sequence period and the number of maximum active users are both given.
In addition to the aforementioned deterministic access schemes, other combinatorial designs for applications in communications, cryptography, and networking can be found in \cite{CDSLP99}.

\subsection{System Model}


We consider a feedback-free multiple-access channel shared by $M$ users, all of them may be active simultaneously, transmitting to a single receiver. 
This model is applicable to bursty traffic.
Channel time is assumed to be divided into time slots of equal duration. 
Each active user reads out the sequence entry from the assigned binary $(0,1)$ sequence sequentially and transmits a packet within a time slot if and only if the sequence value is equal to 1.

To model a time slotted system, time indices are in units of one time slot duration. 
Due to propagation delay, user mobility or random traffic, there are relative time offsets $\tau_i$ of user $i$ for $i=1,2,\ldots,M$, such that a packet from user $i$, received at the time instant $t$ on the receiver's clock, was actually sent at the time instant $t-\tau_i$ on user $i$'s clock. 
These relative time offsets are random, always unknown to the users, but unchanged in a communication session. 
As introduced in \cite{Massey85,CCS06}, there are two different levels of channel synchronization:
\begin{enumerate}[(i)]
\item The channel is \emph{slot-synchronous} if $\tau_i$ is an arbitrary integer for all $i$, i.e., all users know the slot boundaries and transmit packets aligned to the slot boundaries.
\item The channel is \emph{completely asynchronous} if $\tau_i$ is an arbitrary real number for all $i$.
\end{enumerate}
For practical considerations, the slot-synchronous assumption is valid if the synchronization is provided by simple narrow band beacon signals. 
In these scenarios, the receiver can only approximate $\tau_i$ by $\tau_i$ (modulo $1/f$) for all $i$, where $f$ is the used frequency of the beacon signal. 
As such, for simplicity, we restrict our attention to the slot-synchronized model.

If exactly one user transmits a packet within a slot, the packet can be received correctly.
A collision occurs if two or more than two users transmit simultaneously, and all time-overlapping packets are assumed unrecoverable.

For $i=1,2,\ldots,M$, let $\seq{s}_i:= [ s_i(0)\ s_i(1)\ \ldots\ s_i(L-1)]$ be a binary protocol sequence with sequence period (or length) $L$ assigned to user $i$.
Let $\mathbb{Z}_L = \{0,1,2,\ldots,L-1\}$ denote the ring of residues modulo $L$.
Given $\seq{s}_i$, we define the {\em characteristic set} of $\seq{s}_i$ by $\mathcal{I}_i := \{t \in \mathbb{Z}_L : s_i(t) = 1\}$.
We also call $\mathcal{I}_i$ as a ``sequence'', although it is actually represented as a subset of $\mathbb{Z}_L$.
The cardinality of $\I_i$, $|\I_i|$, is called the {\em Hamming weight} of $\I_i$ or $\seq{s}_i$.
Let $$\I_i+\tau_i:=\{k+\tau_i:\,k\in\I_i\},$$ where the addition is performed in $\mathbb{Z}_L$, be the shifted version of $\I_i$ by a \emph{relative shift} $\tau_i$.
If the relative time offset of user $i$ is $\tau_i$, he or she transmits a packet at time slot $t$ if and only if $t\in\I_i+\tau_i$ in modulo $L$.
This paper assumes all sequences have the same Hamming weight $w$ (called \emph{constant-weight}) and share the same period $L$.

For the sake of convenience, for any positive integer $M$, let $[M]:=\{1,2,\ldots,M\}$ and $[M]_i:=[M]\setminus\{i\}$.
Obviously, $[M]=[M]_i$ whenever $i\notin [M]$.
Let $\C=\{\I_1,\I_2,\ldots,\I_M\}$ be a collection of $M$ subsets in $\mathbb{Z}_L$.
$\I_i$ is said to be \emph{unblocked} in $\C$ if for any integer pattern $(\tau_j\in\mathbb{Z}_L:\, j\in[M]_i)$, one has
\begin{equation}
\I_i \nsubseteq \bigcup_{j\in[M]_i} (\I_j + \tau_j).
\label{eq:UI}
\end{equation}
We say $\mathcal{C}$ is \emph{user-irrepressible} (UI) if $\I_i$ is unblocked for all $i\in[M]$.

\begin{example}\label{ex:UI_35}
One can check that the following $\C=\{\I_1,\I_2,\I_3,\I_4\}$ is a UI sequence set of period $L=35$ by~\eqref{eq:UI}.
\begin{align*}
  \I_1&= \{0,10,15,25,30\};  &\I_2= \{0,4,13,17,26\}; \\
  \I_3&= \{0,8,16,24,32\};   &\I_4= \{0,6,12,18,24\}.
\end{align*}
\end{example}

\smallskip

Given two characteristic sets $\I_1$, $\I_2$ and a relative shift $\tau\in\mathbb{Z}_L$ between them.
Let $H_{\I_1\I_2}(\tau)$ denote the \emph{Hamming cross-correlation} between $\I_1$ and $\I_2$ with respect to $\tau$ by giving
\[
H_{\I_1\I_2}(\tau):= |\I_1 \cap (\I_2+\tau)|.
\]
The \emph{maximum Hamming cross-correlation} between $\I_1$ and $\I_2$ is defined as
\[
H_{\I_1\I_2}:=\max_{\tau\in\mathbb{Z}_L} H_{\I_1\I_2}(\tau).
\]
By symmetry, one can see that $H_{\I_1,\I_2}=H_{\I_2\I_1}$.
Let $\lambda_{c}$ be the maximum Hamming cross-correlation for any pair of distinct characteristic sets in $\mathcal{C}$.
We note that $\lambda_c$ measures the maximal mutual interference between any pair of the users in a slot-synchronous channel.
In Example~\ref{ex:UI_35}, one can check that $H_{\mathcal{I}_1\mathcal{I}_{2}}=H_{\mathcal{I}_1\mathcal{I}_{3}}=H_{\mathcal{I}_1\mathcal{I}_{4}}=H_{\mathcal{I}_2\mathcal{I}_{3}}=1$, and $H_{\mathcal{I}_2\mathcal{I}_{4}}=H_{\mathcal{I}_3\mathcal{I}_{4}}=2$.
Hence we have $\lambda_c=2$.

We remark here that for practical considerations, one would like to remove the slot-synchronous assumption. 
It is, in fact, possible to do so and to allow the users to be completely asynchronous.
In a completely asynchronous channel \cite{Hui84}, compared with the slot-synchronous one, the relative time offsets are arbitrary real numbers in units of one time slot duration.
In such a channel, a collision can occur due to partial overlapping of packets.
In other words, a packet is assumed to be successfully sent out if and only if it is not completely or partially overlapped by any other packet.
A sequence set $\C$ with a common period is said to be \emph{completely irrepressible} (CI) \cite{ZSW11} if each user can successfully send out at least one packet per period, no matter what the real offsets are.
By definition, a CI sequence set must be UI.
Conversely, a UI sequence set can be easily modified to provide the CI property at the cost of doubling the sequence period, by padding an extra zero after each sequence entry \cite{Massey85,CCS06,ZSW11}. 
This approach ensures that the maximal number of conflicts occurring in any group of users remains unchanged when the system synchronism is reduced from the slot-synchronous to completely asynchronous. 
As such, we focus on the construction of UI sequences in this paper, and the generalization to the CI sequences is not considered.

\subsection{Related Works and Motivation}

There are various known UI sequences in literature: \emph{Shift-invariant Sequences} (SI)\cite{Massey85,CWS08,SCSW}, \emph{Wobbling Sequences} (WS)\cite{Wong07}, \emph{Extended Prime Sequences} (EPS)\cite{YK95}, \emph{the Chinese Reminder Theorem Sequences} (CRT)\cite{SWSC,CRT} and \emph{CRT Sequences in Prime Users} (CRTp)\cite{UIS09}.
Table~\ref{table:UI_comparison} lists the major parameters of these UI sequences.

\begin{table}[h]
\caption{A comparison of UI sequences designs. $u_M$ is the smallest prime that is larger than or equal to $M$, while $p_M$ is the smallest prime that is strictly larger than $M$.}
\label{table:UI_comparison}
\scriptsize
\begin{tabular}{c||l|l|l|l|l}
\hline
Construction & $\begin{array}{l} \text{Applicable} \\ \text{User Number} \end{array}$ & $\begin{array}{l} \text{Constant-} \\ \text{Weight?} \end{array}$ & $\begin{array}{l} \text{Hamming} \\ \text{Weight} \end{array}$ & $\begin{array}{l} \text{Sequence} \\ \text{Period} \end{array}$ & $\lambda_c$ \\[2pt] \hline 
SI \cite{Massey85,CWS08,SCSW} & positive integer $M$ & Yes & $2^{M-1}$ & $2^M$ & $2^{M-2}$ \\[2pt] 
WS \cite{Wong07} & odd prime $p$ & Yes & $p^3$ & $p^4$ & $p^2$ \\[2pt]
EPS \cite{YK95} & odd prime $p$ & Yes & $p$ & $p(2p-1)$ & $1$ \\[2pt] 
CRT \cite{SWSC,CRT} & positive integer $M$ & Yes & $M$ & $u_M(2M-1)$ & $1$ \\[2pt] 
CRTp \cite{UIS09} & odd prime $p$ & No & $\begin{array}{l} \text{one for }p, \text{ and} \\ \text{others for } p+1 \end{array}$ & $2p(p-1)$ & $2$ \\[10pt] 
CRTm & positive integer $M$ & Yes & $M+1$ & $p_M(2M-1)$ & $2$ \\ \hline
\end{tabular}
\end{table}

The primary design objective of UI sequences is to minimize the sequence period $L$ when $M$ is given, as $L$ has a fundamental impact on the worst-case channel access delay, i.e., the maximum waiting time that a message can be successfully received.
Let $L_{min}(M)$ be the smallest $L$ such that a set of $M$ UI sequences of common period $L$ exists.
The work in \cite{UIS09} shows that $L_{min}(M)$ is lower bounded by $8M^2/9$.
One can see from Table~\ref{table:UI_comparison} that SI sequences are the shortest known UI sequences for $M\leq 6$, CRTp sequences are the shortest for all prime $M\geq 7$, and CRT sequences are the shortest for all non-prime $M\geq 8$.
Moreover, the work in \cite{SWSC} improves the lower bound on the sequence period from $8M^2/9$ to $2M^2$ for the case with constant Hamming weight $w=M$, and shows that the CRT sequences achieve this lower bound asymptotically.

We now know that CRT sequences are the shortest known UI sequences for all non-prime $M\geq 8$, however, their small Hamming weight would possibly yield larger average delay~\cite{VANET13,VANET14}, which is also an important considered metric in the evaluation of channel access schemes.
This observation motivates us to investigate short UI sequences with $w>M$.
In this paper, we propose a CRTm construction of UI sequences, which is also based on the \emph{Chinese Remainder Theorem}~\cite{Ireland90}.

\subsection{Contribution}\label{sec:intro_contri}

Our proposed CRTm sequences are with period $p_M(2M-1)$, constant Hamming weight $M+1$ and $\lambda_c=2$, where $p_M$ is the smallest prime that is larger than $M$, for any positive integer $M$.
Obviously, $p_M=u_M$ when $M$ is a non-prime.
Both CRT sequences and CRTm sequences are the shortest known UI sequences for all non-prime $M\geq 8$, but CRTm sequences have a larger Hamming weight.
It will be shown that this larger Hamming weight would bring better average individual-delay and group-delay.
In addition, similar to the improvement of minimum sequence period for $w=M$ in \cite{SWSC}, we derive an asymptotic lower bound of $2M^2$ for $w=M+1$ if the sequence structure satisfies some technical conditions, called {\em equi-difference}, and hence prove the CRTm construction is optimal in the sense that it can achieve this lower bound.
Our method can be viewed as a generalization of that in~\cite{SWSC} for $w=M$.

It is worth pointing out that our proposed CRTm construction allows $w\leq(M-1)\lambda_c$, whereas all other known constant-weight UI sequences except SI in Table~\ref{table:UI_comparison} strictly require $w>(M-1)\lambda_c$.
This latter condition clearly implies the UI property.
The difference on the relation between $w$ and $\lambda_c$ makes the proof for the UI property of the CRTm construction very different and more complicated.

The rest of this paper is organized as follows.
After setting up some definitions and notation in Section~\ref{sec:definition}, we provide the CRTm construction of UI sequences with constant weight $w=M+1$ in Section~\ref{sec:CRTm}.
Then, in Section~\ref{sec:optimal}, we establishe an asymptotic lower bound on the minimum sequence period with $w=M+1$, and prove that the CRTm construction can achieve this lower bound if the sequences are constant-weight and equi-difference.
Section~\ref{sec:simulation} is devoted to demonstrate the delay performance of CRTm sequences through numerical study.
A conclusion is given in Section~\ref{sec:conclusion}.


\section{Definitions and Notation}\label{sec:definition}
Given a sequence set $\C=\{\mathcal{I}_1,\mathcal{I}_2,\ldots, \mathcal{I}_M\}$.
For $i=1,2,\ldots,M$, let $\set{B}_i$ be a collection of all indices $k$ ($k\neq i$) such that the $H_{\I_i\I_k}=\max_{j\in[M]_i}H_{\I_i\I_j}$, that is,
\[
 \set{B}_i:=\left\{k\in[M]_i:\, H_{\mathcal{I}_i\mathcal{I}_{k}}\geq H_{\mathcal{I}_i\mathcal{I}_j}, \forall j\in[M]_i\right\}.
\]
Let $\set{T}_{i,k}$ be a collection of all relative shifts such that
\[
 \set{T}_{i,k}=\left\{\tau_k\in\mathbb{Z}_L:\, H_{\I_i\I_k}(\tau_k)= H_{\I_i\I_k} \right\}.
\]
For $\tau_k$ in $\set{T}_{i,k}$, we define
\[
\mathcal{I}_{i,k,\tau_k} :=\mathcal{I}_{i} \setminus (\mathcal{I}_{k}+\tau_k).
\]

In Example~\ref{ex:UI_35}, since $H_{\mathcal{I}_1\mathcal{I}_{2}}=H_{\mathcal{I}_1\mathcal{I}_{3}}=H_{\mathcal{I}_1\mathcal{I}_{4}}=H_{\mathcal{I}_2\mathcal{I}_{3}}=1$, and $H_{\mathcal{I}_2\mathcal{I}_{4}}=H_{\mathcal{I}_3\mathcal{I}_{4}}=2$, we have $\set{B}_1=\{2,3,4\}$, $\set{B}_2=\set{B}_3=\{4\}$ and $\set{B}_4=\{2,3\}$.
Meanwhile, $\set{T}_{3,4}=\{0,8\}$, $\mathcal{I}_{3,4,0}=\{8,16,32\}$, and $\mathcal{I}_{3,4,8}=\{0,16,24\}$.

Let $\I$ be a sequence of period $L$ and Hamming weight $w$.
We use $d^*(\set{I})$ to denote the {\em set of non-zero differences} between pairs of distinct elements in~$\set{I}$, i.e.,
\[
 d^*(\set{I}) := \{a-b ~(\text{mod } L):\, a,b\in \set{I},a\neq b\}.
\]
Obviously, $|d^*(\set{I})|\geq |\set{I}|-1$.
$\set{I}$ is said to be \emph{exceptional} if $|d^*(\set{I})|< 2|\set{I}|-2$.
$\set{I}$ is called {\em equi-difference} if the elements in $\set{I}$ form an arithmetic progression in $\mathbb{Z}_L$, i.e.,
\[
 \set{I} =\left\{0, g, 2g, \ldots, (w-1)g \right\} \quad \text{for some }g\in\mathbb{Z}_L,
\]
where the product is performed in $\mathbb{Z}_L$.
The element $g$ is called a {\em generator} of $\set{I}$.
If all $\mathcal{I}$s in $\mathcal{C}$ are equi-difference, then we say $\mathcal{C}$ is equi-difference.

In Example~\ref{ex:UI_35}, $\I_1$, $\I_2$, $\I_3$ and $\I_4$ are equi-difference with Hamming weight $w=5$ and generators $g_1=15$, $g_2=13$, $g_3=8$ and $g_4=6$, respectively.
One can further check that $d^*(\set{I}_1)=\{5,10,15,20,25,30\}$ and thus $\set{I}_1$ is exceptional.


\section{A New Construction of UI Sequences}\label{sec:CRTm}

We start this section with a necessary and sufficient condition for a sequence set to be UI.

\begin{lemma}
Let $\C=\{\I_1,\I_2,\ldots,\I_M\}$ be a sequence set.
$\C$ is UI if and only if, for $i\in[M]$, $A\subseteq[M]_i$ and arbitrary relative shifts $\tau_j$, $j\in A$, one has
\begin{equation}
\left| \I_i\setminus\bigcup_{j\in A}(\I_j+\tau_j) \right| \geq M-|A|.
\label{eq:UI_generalized}
\end{equation}
\label{prop:UI}
\end{lemma}
\begin{proof}
We only consider the necessary part as the sufficient part is simply obtained by letting $A=[M]_i$.
Assume to the contradiction that there exist $i\in[M]$, $A\subseteq[M]_i$ and $\tau_j$, $j\in A$ such that the cardinality of $\I_i\setminus\bigcup_{j\in A}(\I_j+\tau_j)$ is at most $M-|A|-1$.
Then we always can choose some relative shifts $\tau_k$ for all $k\in[M]_i\setminus A$ such that 
\begin{equation}\label{eq:UI_equivalent}
\I_i\setminus\bigcup_{j\in A}(\I_j+\tau_j)\subseteq\bigcup_{k\in[M]_i\setminus A}(\I_k+\tau_k).
\end{equation}
More precisely, one can iteratively cover one element in the left hand side of \eqref{eq:UI_equivalent} by a sequence $\I_k$ or its shifted version by $\tau_k$ for some $k\in[M]_i\setminus A$.
This contradicts to \eqref{eq:UI}.
\qed
\end{proof}

Now, we present a new construction of constant-weight UI sequences, called the CRTm construction, which is a variation of the CRT construction~\cite{SWSC}.

Let $p$ and $q$ be relatively prime integers.
Let $\mathbb{Z}_p\otimes\mathbb{Z}_q$, consisting of all ordered pairs $(a,b)$ with $a\in\mathbb{Z}_p$ and $b\in\mathbb{Z}_q$, be the direct product of the rings $\mathbb{Z}_p$ and $\mathbb{Z}_q$.
There is a natural bijection (ring isomorphism) $f:\mathbb{Z}_{pq}\to\mathbb{Z}_p\otimes\mathbb{Z}_q$ by defining
\[
f(x) := (x \bmod p, x \bmod q),
\]
that is, $f$ preserves addition and multiplication on both sides.
If there is no danger of confusion, operations under $\mathbb{Z}_p\otimes\mathbb{Z}_q$ are componentwise taken modulo $p$ and $q$.
We will construct sequences by specifying characteristic sets in $\mathbb{Z}_p \otimes \mathbb{Z}_q$.

\medskip

\noindent{\bf CRTm Construction:}
Given $M \geq 4$, we set $p_M$ to be the smallest prime larger than $M$.
Let
\[
\widehat{\I}_j :=
\begin{cases}
   \Big\{ (jy,y) \in \mathbb{Z}_{p_M} \otimes \mathbb{Z}_{2M-1}:\, y = 0,1,\ldots, M\Big\}, &j=0,1,\ldots,p_M-1; \\
  \Big\{ (y,0) \in \mathbb{Z}_{p_M} \otimes \mathbb{Z}_{2M-1}:\, y = 0,1,\ldots, M\Big\}, \ &j=p_M.
\end{cases}
\]
Notice that $p_M$ is relatively prime to $2M-1$ due to $M\geq 4$ and the Bertrand's postulate, which states that there is always a prime strictly between $M$ and $2M-2$ for any integer $M\geq 4$.
Then we obtain the characteristic sets of the sequences, $\mathcal{I}_j$, by taking the inverse image $f^{-1}(\widehat{\I}_j)$ for $j=0,\ldots, {p_M}$.
The CRTm construction produces $p_M+1$ sequences of length ${p_M}(2M-1)$ and constant-weight $M+1$.

\begin{example}\label{ex:UI_77}
Given $M=6$, the CRTm construction produces 8 sequences of period 77 and constant-weight 7. The characteristic sets are:
\begin{align*}
  \set{I}_0= \{0,56,35,14,70,49,28\}; \ \ \ \ &\set{I}_1= \{0,1,2,3,4,5,6\}; \\
  \set{I}_2= \{0,23,46,69,15,38,61\}; \ \ \ \ &\set{I}_3= \{0,45,13,58,26,71,39\};\\
  \set{I}_4= \{0,67,57,47,37,27,17\}; \ \ \ \ &\set{I}_5= \{0,12,24,36,48,60,72\};\\
  \set{I}_6= \{0,34,68,25,59,16,50\}; \ \ \ \ &\set{I}_7= \{0,11,22,33,44,55,66\}.
\end{align*}
\end{example}

\begin{remark}
The previously known CRT construction can be obtained from the CRTm construction by removing the last ordered pair from each sequence, namely, by letting $y=0,1,\ldots,M-1$.
As the CRT construction allows $\lambda_c=1$ \cite{SWSC,CRT}, the UI property can be derived directly by $w>(M-1)\lambda_c$.
However, one can see from Example~\ref{ex:UI_77} that $\lambda_c\geq 2$ for CRTm construction. 
Even though the two constructions look similar, the proof of the UI property is very different as the CRTm construction allows $w\leq (M-1)\lambda_c$.
\end{remark}


Following Example~\ref{ex:UI_77}, we consider the set $\C=\{\I_0,\I_1,\I_3,\I_4,\I_5,\I_7\}$, and try to show that $\I_1$ is unblocked in $\C$.
Observe that $H_{\I_1,\I_3}=2$, $H_{\I_1,\I_5}=2$, and $H_{\I_1,\I_0}=H_{\I_1,\I_4}=H_{\I_1,\I_7}=1$.
It is easy to see that $\set{T}_{1,3}=\{6\}$, that is, only the relative shift $\tau=6$ satisfies $H_{\I_1,\I_3}(\tau)=2$.
Let $\I':=\I_1\setminus(\I_3+6)=\{1,2,3,4,5\}$.
One has $H_{\I',\I_5}=1$, which implies that at most $3$ packets in $\I_1$ can be blocked by $\I_3$ and $\I_5$ simultaneously. 
Since at most one packet in $\I_1$ can be blocked by each of the other three sequences, $\I_1$ is unblocked in $\C$.

The above example illustrates the idea of how to prove a set of sequences having $w\leq(M-1)\lambda_c$ to be UI.
As for the CRTm construction, we aim to show that any $M$ sequences obtained by the construction form a UI sequence set.
We start with the following lemma that gives an equivalent condition for the existence of UI sequences with constant-weight $w=M+1$.

\begin{lemma}
A sequence set $\mathcal{C}=\{\mathcal{I}_1, \mathcal{I}_2, \ldots, \mathcal{I}_M\}$ of constant-weight $M+1$ is UI if and only if, for $i\in[M]$, one has
\begin{enumerate}
\item $H_{\mathcal{I}_{i}\mathcal{I}_{j}}\leq 2$ for any $j\in[M]_i$; and
\item $H_{\mathcal{I}_{i,k,\tau_k}\mathcal{I}_{j}} \leq  1$, i.e., $d^*(\mathcal{I}_{i,k,\tau_k}) \cap d^*(\mathcal{I}_j)=\emptyset$, for any distinct integers $j, k\in[M]_i$ such that $k\in\set{B}_i$ and any $\tau_k \in \set{T}_{i,k}$.
\end{enumerate}
\label{lem:eqc}
\end{lemma}
\begin{proof}
First, we prove the necessary part by contradiction.

(i) Suppose $H_{\mathcal{I}_{i}\mathcal{I}_{j}}\geq 3$ for some $j\in[M]_i$.
Then $|\I_i\setminus(\I_j+\tau_j)|\leq M-2<M-1$ for some $\tau_j\in\set{T}_{i,j}$, which contradicts to \eqref{eq:UI_generalized} by setting $A=\{j\}$.

(ii) Suppose $H_{\I_{i,k,\tau_k}\I_{j}} \geq  2$ for some distinct integers $j, k\in[M]_i$ such that $k\in \set{B}_i$ and some $\tau_k\in\set{T}_{i,k}$.
By the defining property of $\I_{i,k,\tau_k}$, we have $\I_{i,k,\tau_k}\subset\I_i$.
It is easy to see $H_{\I_i\I_j} \geq H_{\I_{i,k,\tau_k}\I_j} \geq 2$.
Since $k\in\set{B}_i$, it further implies that $H_{\I_i\I_k}\geq H_{\I_i\I_j}\geq 2$.
Let $\tau_j$ be the relative shift so that $H_{\I_{i,k,\tau_k}\I_{j}}(\tau_j)=H_{\I_{i,k,\tau_k}\I_{j}}$.
Then,
\begin{align*}
\left|\I_i\setminus\Big((\I_k+\tau_k)\cup(\I_j+\tau_j)\Big) \right| &= \left|\Big(\I_i\setminus(\I_k+\tau_k)\Big)\setminus(\I_j+\tau_j) \right| \\
&\leq \left|\Big(\I_i\setminus(\I_k+\tau_k)\Big) \right| -2 \\
&\leq (M+1)-2-2 < M-2,
\end{align*}
which contradicts to \eqref{eq:UI_generalized} by setting $A=\{j,k\}$.

For the sufficient part, with condition (i) and (ii), as $w=M+1$, it is easy to see that $\mathcal{I}_i \nsubseteq \cup_{j\in[M]_i} (\mathcal{I}_j + \tau_j)$ for any $i$ and any relative shifts $\tau_j$.
It implies that $\mathcal{C}$ is UI.
\qed
\end{proof}

We are ready to prove the UI property of the CRTm construction.
\begin{theorem}
Any $M$ sequences from the CRTm construction form an equi-difference UI sequence set of length ${p_M}(2M-1)$ and constant-weight $M+1$.
\label{thm:crtm}
\end{theorem}
\begin{proof}
Observe that $\widehat{\I}_j$ has the generator $(j,1)$ for $j=0,1,\ldots, p_M-1$ and generator $(1,0)$ when $j=p_M$.
Thus, the CRTm construction produces $p_M+1$ equi-difference sequences of length $p_M(2M-1)$.
We define $d^*(\widehat{\I}_j)$ in the same way as $d^*(\mathcal{I}_j)$, but with the addition and subtraction done in $\mathbb{Z}_{p_M} \otimes \mathbb{Z}_{2M-1}$ instead of $\mathbb{Z}_{p_M(2M-1)}$.
It is sufficient to show that each obtained sequence $\widehat{\I}_i$, $i\in\{0,1,\ldots,p_M\}$, satisfies the two conditions in Lemma~\ref{lem:eqc}.
Note that, if $H_{\widehat{\I}_i\widehat{\I}_j}=1$ for any $j\neq i$, then both the two conditions of Lemma~\ref{lem:eqc} hold for $i$.

First, consider $i=p_M$.
We claim that $H_{\widehat{\I}_{p_M}\widehat{\I}_j}=1$.
Suppose to the contradiction that $H_{\widehat{\I}_{p_M}\widehat{\I}_j}\geq 1$, that is, $d^*(\widehat{\I}_{p_M})\cap d^*(\widehat{\I}_j)\neq\emptyset$, for some $j\in\mathbb{Z}_{p_M}$.
Then
\[
(y_1,0)-(y_1',0) \equiv (jy_2,y_2) - (jy_2', y_2') \quad \text{in } \mathbb{Z}_{p_M}\otimes\mathbb{Z}_{2M-1},
\]
for $y_1,y_1',y_2,y_2'\in\{0,1,\ldots,M\}$ with $y_1\neq y_1'$ and $y_2\neq y_2'$.
By equating the second components on both sides, we have $y_2=y_2'$, a contradiction to $y_2\neq y_2'$.
Hence $\mathcal{I}_{p_M}$ possesses the two conditions in Lemma~\ref{lem:eqc}.
\medskip

Second, consider $i\in\mathbb{Z}_{p_M}$.
If $H_{\widehat{\I}_i\widehat{\I}_j}=1$ for any $j\neq i$, then we are done; otherwise, $H_{\widehat{\I}_i\widehat{\I}_j}\geq 2$ for some $j\in\mathbb{Z}_{p_M}\setminus\{i\}$, i.e., $d^*(\widehat{\I}_i)\cap d^*(\widehat{\I}_j)\neq\emptyset$.
Note that $p_M$ is not a candidate for $j$ due to the first part.
$d^*(\widehat{\I}_i)\cap d^*(\widehat{\I}_j)\neq\emptyset$ implies that 
\[
(iy_1, y_1) - (iy_1', y_1') \equiv (jy_2, y_2) - (jy_2', y_2') \quad \text{in } \mathbb{Z}_{p_M}\otimes\mathbb{Z}_{2M-1},
\]
for some $y_1,y_1',y_2,y_2'\in\{0,1,\ldots,M\}$ with $y_1\neq y_1'$ and $y_2\neq y_2'$.
By equating the second components on both sides, we have $y_1-y_1' \equiv y_2-y_2'$ (mod $2M-1$).
Since $0 \leq y_1,y_1',y_2,y_2' \leq M$, there are only five possible solutions to $y_1-y_1'$ and $y_2-y_2'$, as follows.
\begin{align}
y_1-y_1'&=y_2-y_2';  \label{eq:c01}\\
y_1-y_1'=M, & \ \ \  y_2-y_2'=-(M-1); \label{eq:c11}\\
y_1-y_1'=-M, & \ \ \ y_2-y_2'=M-1; \label{eq:c12}\\
y_1-y_1'=-(M-1), & \ \ \ y_2-y_2'=M; \label{eq:c02} \\
y_1-y_1'=M-1, & \ \ \ y_2-y_2'=-M. \label{eq:c03}
\end{align}
If $y_1-y_1'=y_2-y_2'$, from the first component, we have $(i-j) (y_1-y_1') \equiv 0$ (mod $p_M$), which implies that $i=j$ or $y_1=y_1'$ due to $i,j\in\mathbb{Z}_{p_M}$ and $0\leq y_1,y_1'\leq M$.
This contradicts the assumption that $i\neq j$ and $y_1\neq y_1'$.
Hence \eqref{eq:c01} can be excluded.
The remaining four possible solutions \eqref{eq:c11}--\eqref{eq:c03} imply respectively \eqref{eq:solution_1}--\eqref{eq:solution_4}.
\begin{align}
(y_1,y_1')&=(M,0) &\text{and} \quad (y_2,y_2')&=(0,M-1) \text{ or } (1,M); \label{eq:solution_1} \\
(y_1,y_1')&=(0,M) &\text{and} \quad (y_2,y_2')&=(M-1,0) \text{ or } (M,1); \label{eq:solution_2} \\
(y_1,y_1')&=(0,M-1) \text{ or } (1,M) &\text{and} \quad (y_2,y_2')&=(M,0); \label{eq:solution_3} \\
(y_1,y_1')&=(M-1,0) \text{ or } (M,1) &\text{and} \quad (y_2,y_2')&=(0,M). \label{eq:solution_4}
\end{align}

Intuitively speaking, from the construction of $$\widehat{\I}_j=\Big\{\big(0,0\big), \big(j,1\big), \big(2j,2\big), \ldots,\big(j(M-1),M-1\big),\big(jM,M\big)\Big\},$$ we call elements $(0,0)$, $(j,1)$, $(j(M-1),M-1)$ and $(jM,M)$ the head, second-head, second-tail and tail, respectively.
The above arguments conclude the following property, say \emph{collided property}.
\begin{itemize}
\item[] If there is a pair of repeated elements between two sequences under some relative shift, the two elements must be $\{\text{head,~tail}\}$ of one sequence, and $\{\text{head,~second-tail}\}$ or $\{\text{second-head,~tail}\}$ of another.
\end{itemize}
Back to $\widehat{\I}_i$, the collided property immediately implies $H_{\widehat{\I}_i\widehat{\I}_j}\leq 2$ for any $j\neq i$, i.e., Lemma~\ref{lem:eqc}(i).
Now, it remains to show that Lemma~\ref{lem:eqc}(ii) holds as well.
Consider $k\in\mathcal{B}_i$.
If $H_{\widehat{\I}_i\widehat{\I}_k}=1$, then the result follows.
If $H_{\widehat{\I}_i\widehat{\I}_k}\geq 2$ and $H_{\widehat{\I}_{i,k,\tau_k}\widehat{\I}_j}\geq 2$, by the definition of $\widehat{\I}_{i,k,\tau_k}$ and condition (i), we have  $H_{\widehat{\I}_i\widehat{\I}_k}=2$, $H_{\widehat{\I}_i\widehat{\I}_j}=2$ and $H_{\widehat{\I}_{i,k,\tau_k}\widehat{\I}_j}=2$.
By the collided property, both the repeated elements of $\I_k$ and $\I_j$ are in the form $\{\text{head,~tail}\}$, while the repeated elements of $\I_i$ with respect to $\I_k$ and $\I_j$ are in the form $\{\text{head,~second-tail}\}$ (or $\{\text{second-head,~tail}\}$) and $\{\text{second-head,~tail}\}$ (or $\{\text{head,~second-tail}\}$).
Since in a single sequence $\{\text{second-head,~tail}\}$ is a shifted version of $\{\text{head,~second-tail}\}$, the two $\{\text{head,~tail}\}$ pairs of $\I_k$ and $\I_j$ are repeated under some relative shift, which contradicts to the collided property.
Hence we complete the proof.
\qed
\end{proof}


\section{A Tight Asymptotic Lower Bound on Sequence Period }\label{sec:optimal}
In this section, we derive an asymptotic lower bound on sequence period for the equi-difference structure and $w=M+1$, and then show this lower bound is tight by using the CRTm construction.

\subsection{Preliminaries}
We start with three basic results which will be useful to derive a lower bound on sequence period.

The following necessary condition for $\mathcal{C}=\{\mathcal{I}_1, \mathcal{I}_2, \ldots, \mathcal{I}_M\}$ of constant-weight $M+1$ to be UI directly follows from Lemma~\ref{lem:eqc}(ii), and thus its proof is omitted here.

\begin{proposition}
Let $\mathcal{C}=\{\mathcal{I}_1, \mathcal{I}_2, \ldots, \mathcal{I}_M\}$ be a UI sequence set of constant-weight $M+1$.
Then for $i\neq j\in[M]$, one has
\[
H_{\mathcal{I}_{i,k,\tau_k}\mathcal{I}_{j,l,\tau_l}}\leq 1, \text{ i.e., } d^*(\set{I}_{i,k,\tau_k}) \cap d^*(\set{I}_{j,l,\tau_l}) = \emptyset
\]
for $k\in\set{B}_i$, $\tau_k\in\set{T}_{i,k}$, $l\in\set{B}_j$, $\tau_l\in\set{T}_{j,l}$ and $\{i,k\}\neq\{j,l\}$.
\label{prop:nec}
\end{proposition}

The following result provides an upper bound on $|d^*(\I_j) \cap d^*(\I_j)|$ for $i\neq j$.

\begin{lemma}
Let $\I_1$ and $\I_2$ be two equi-difference sequences of the same period.
If $H_{\mathcal{I}_1\mathcal{I}_2}\leq2$, then
\[
\big|d^*(\set{I}_1) \cap d^*(\set{I}_2)\big| \leq 4.
\]
\label{lem:equi}
\end{lemma}
\begin{proof}
Let $g_1$ and $g_2$ are the generators of $\I_1$ and $\I_2$, respectively.
Then $\I_1$ and $\I_2$ are of the form:
\[
 \I_1 = \{0,g_1,2g_1,\ldots,(w-1)g_1\}; \ \  \I_2 = \{0,g_2,2g_2,\ldots,(w-1)g_2\}.
\]
Consider the following four sets:
\begin{align*}
d_a^*(\I_{1})&= \{g_1, 2g_1, \ldots, (w-1)g_1\}; \\
d_b^*(\I_{1})&= \{-g_1, -2g_1, \ldots, -(w-1)g_1\}; \\
d_a^*(\I_{2})&= \{g_2, 2g_2, \ldots, (w-1)g_2\};   \\
d_b^*(\I_{2})&= \{-g_2, -2g_2, \ldots, -(w-1)g_2\}.
\end{align*}
Obviously, $d_a^*(\I_1)\cup d_b^*(\I_1)=d^*(\I_1)$ and $d_a^*(\I_2)\cup d_b^*(\I_2)=d^*(\I_2)$.

Now we prove this lemma by contradiction.
Suppose $|d^*({\set{I}_1}) \cap d^*({\set{I}_2})| \geq 5$.
Then we have
\begin{align*}
5 &\leq |d^*({\I_1})\cap d^*({\I_2})| = \Big|\big(d_a^*(\I_{1})\cup d_b^*(\I_{1})\big) \cap \big(d_a^*(\I_{2})\cup d_b^*(\I_{2})\big)\Big| \\
&= \left|\Big(d_a^*(\I_{1})\cap d_a^*(\I_{2})\Big) \cup \Big(d_a^*(\I_{1})\cup d_b^*(\I_{2})\Big) \cup \Big(d_b^*(\I_{1})\cap d_a^*(\I_{2})\Big) \cup \Big(d_b^*(\I_{1})\cup d_b^*(\I_{2})\Big)\right|
\end{align*}
which implies that at least one of the four sets $d_a^*(\I_{1})\cap d_a^*(\I_{2})$, $d_a^*(\I_{1})\cup d_b^*(\I_{2})$, $d_b^*(\I_{1})\cap d_a^*(\I_{2})$ and $d_b^*(\I_{1})\cup d_b^*(\I_{2})$ has cardinality at least 2.
There are four cases as follows.
\begin{enumerate}[Case 1.]
\item If $|\big(d_a^*(\I_{1})\cap d_a^*(\I_{2})\big)|\geq 2$, then $H_{\I_1\I_2}(0)\geq 3$.
\item If $|\big(d_a^*(\I_{1})\cap d_b^*(\I_{2})\big)|\geq 2$, then $H_{\I_1\I_2}(-(w-1)g_2)\geq 3$.
\item If $|\big(d_b^*(\I_{1})\cap d_a^*(\I_{2})\big)|\geq 2$, then $H_{\I_1\I_2}((w-1)g_1)\geq 3$.
\item If $|\big(d_b^*(\I_{1})\cap d_b^*(\I_{2})\big)|\geq 2$, then $H_{\I_1\I_2}(0)\geq 3$.
\end{enumerate}
Each of the four cases implies that $H_{\I_1\I_2}\geq 3$, which contradicts to the assumption that $H_{\I_1\I_2}\leq 2$.
This completes the proof.
\qed
\end{proof}

We illustrate Lemma~\ref{lem:equi} with the following example.

\begin{example}\label{ex:UI_40}
Consider the following equi-difference sequences $\I_1$, $\I_2$ and $\I_3$ of period 40 with generators $g_1=7$, $g_2=9$, $g_3=17$ and Hamming weight $6$:
\begin{align*}
  \I_1&= \{0,7,14,21,28,35\}, \\
  \I_2&= \{0,9,18,27,36,5\}, \\
  \I_3 &= \{0,17,34,11,28,5\}.
\end{align*}
We have $H_{\I_1\I_2}=H_{\I_2\I_3}=H_{\I_1\I_3}=2$ while
$|d^*(\I_1)\cap d^*(\I_2)|=|d^*(\I_2)\cap d^*(\I_3)|=|\{5,35\}| =2$ and $|d^*(\I_1)\cap d^*(\I_3)|=|\{5,12,28,35\}| =4$.
\end{example}

We also need the following previously known result to quantify the maximum number of exceptional sequences if their non-zero difference sets are mutually disjoint.

\begin{lemma}[\cite{SWC10}]
Consider $\set{I}$s in $\mathbb{Z}_L$ whose Hamming weights are all equal to $u$.
Let $\pi(L,2u-2)$ be the number of distinct relatively prime divisors of $L$ between $2$ and $2u-2$.
There are at most $\pi(L,2u-2)$ exceptional $\set{I}$s such that their non-zero difference sets are mutually disjoint.
\label{lemma:shum}
\end{lemma}


\subsection{A lower bound on sequence period}
In order to obtain the minimal number of mutually disjoint non-zero difference sets in a UI sequence set, we set up the following definitions.

For $i\in[M]$, let $r_i$ be the smallest integer in $\set{B}_{i}$.
Define
\[
\set{F}:=\left\{i\in[M] :\, i=r_{r_i} \right\}.
\]
$\set{F}$ can be further divided into the following two disjoint subsets:
\begin{align*}
&\set{F}_1:=\left\{i \in [M] :\, i=r_{r_i}, i>r_i \right\}, \\
&\set{F}_2:=\left\{i \in [M] :\, i=r_{r_i}, i<r_i \right\}.
\end{align*}
Obviously, $|\set{F}_1|=|\set{F}_2|=|\set{F}|/2$.
In Example~\ref{ex:UI_35}, $\set{F}=\{2,4\}$ and $\set{F}_1=\{4\}$, $\set{F}_2=\{2\}$.
The motivation of these definitions will be clear after Theorem~\ref{thm:bound1}.

Given a positive integer $M$, let $L^{e}_{\min}(M)$ be the smallest period $L$ such that an equi-difference UI sequence set with constant-weight $M+1$ exists.
We are ready for our main result in this section.

\begin{theorem}
\[ \liminf_{M\rightarrow\infty} \frac {L^e_{\min}(M)}{2M^2} \geq 1.
\]\label{thm:bound1}
\end{theorem}
\begin{proof}
Consider an equi-difference sequence set $\mathcal{C}=\{\mathcal{I}_1, \mathcal{I}_2, \ldots, \mathcal{I}_M\}$ of constant-weight $M+1$.
For $i=1,2,\ldots,M$, let $\set{I}^*_{i,r_i}$ be a subset of $\set{I}_{i,r_i,\tau^*_{r_i}}$ such that $|\set{I}^*_{i,r_i}|=M-1$, where $\tau^*_{r_i}$ is the smallest element in $\set{T}_{i,r_i}$.

By Proposition~\ref{prop:nec}, the non-zero difference sets of the following distinct $M- |\set{F}| +|\set{F}_1|$ sequences are mutually disjoint:
\begin{equation}
\left\{ \set{I}^*_{i,r_i}:\, i\in\big([M]\setminus\set{F}\big)\cup\set{F}_1 \right\}.
\label{eq:exceptional_1}
\end{equation}
Similarly, the non-zero difference sets of the following distinct $M- |\set{F}|+|\set{F}_2|$ sequences are mutually disjoint:
\begin{equation}
\left\{ \set{I}^*_{i,r_i}:\, i\in\big([M]\setminus\set{F}\big)\cup\set{F}_2 \right\}.
\label{eq:exceptional_2}
\end{equation}
By Lemma~\ref{lemma:shum}, plugging $u=M-1$ implies that there are at most $\pi(L,2M-4)$ exceptional sequences in each of \eqref{eq:exceptional_1} and \eqref{eq:exceptional_2}.
By Proposition~\ref{prop:nec}, we further have
\[
\Big(d^*(\set{I}^*_{i_1,r_{i_1}}) \cup d^*(\set{I}^*_{r_{i_1},i_1})\Big) \bigcap  \Big(d^*(\set{I}^*_{i_2,r_{i_2}}) \cup d^*(\set{I}^*_{r_{i_2},i_2})\Big)  = \emptyset
\]
for any distinct $i_1, i_2 \in \set{F}_1$ or any distinct $i_1, i_2 \in \set{F}_2$.

Since the total number of distinct nonzero differences cannot be larger than the number of nonzeros in $\mathbb{Z}_L$, we have:
\begin{align}
L-1\geq & \sum_{i \in [M]\setminus\set{F}}\Big|d^*(\set{I}^*_{i,r_i})\Big| + \sum_{i \in \set{F}_1} \Big|d^*(\set{I}^*_{i,r_i}) \cup d^*(\set{I}^*_{r_i,i})\Big| \notag \\
        = & \sum_{i \in [M] }\Big|d^*(\set{I}^*_{i,r_i})\Big|- \sum_{i \in \set{F}_1}\Big |d^*(\set{I}^*_{i,r_i}) \cap d^*(\set{I}^*_{r_i,i})\Big|. \label{eq:lower_1}
\end{align}
By Lemma~\ref{lem:equi} and the condition that $H_{\set{I}_i,\set{I}_{r_i}}\leq 2$ for all $i$, we have
\begin{equation}\label{eq:lower_2}
\big|d^*(\set{I}^*_{i,r_i}) \cap d^*(\set{I}^*_{r_i,i})\big|\leq \big|d^*(\set{I}_{i}) \cap d^*(\set{I}_{r_i})\big| \leq 4.
\end{equation}
Combining \eqref{eq:lower_1} and \eqref{eq:lower_2} derives
\begin{align}
L \geq & \sum_{i \in [M] }\Big|d^*(\set{I}^*_{i,r_i})\Big|- 4\Big|\set{F}_1\Big|+1 \notag \\
  \geq & \Big(M-2\pi(L,2M-4)\Big)(2M-4) +  2\pi(L,2M-4)(M-2) - 4|\set{F}_1|+1 \notag \\
  \geq & (2M-4)M-2\pi(L,2M-4)(M-2)-2M +1. \label{eq:lower_3}
\end{align}
The second ``$\geq$'' of the above is due to the fact that there are at most $2\pi(L,2M-4)$ exceptional sequences in total $M$ sequences, while the last one follows from $|\set{F}_1|=|\set{F}|/2\leq M/2$.
Taking lim sup on both sides of \eqref{eq:lower_3} leads to
\begin{align*}
\liminf_{M\rightarrow\infty} \frac {L}{2M^2}  &\geq  \liminf_{M\rightarrow\infty} \frac {(2M-4)M}{2M^2}-\frac {2\pi(L,2M-4)(M-2)}{2M^2}-\frac{2M-1}{2M^2} \\
&= 1- \liminf_{M\rightarrow\infty} \frac{\pi(L,2M-4)}{M} =1.
\end{align*}
The last identity above is due to
\[
\liminf_{M\rightarrow\infty} \frac{\pi(L,2M-4)}{(2M-4) / \ln (2M-4)}\leq 1,
\]
which can be obtained by the prime number theorem.
Hence we complete the proof.
\qed
\end{proof}

By the CRTm construction, we show that the asymptotic lower bound in Theorem~\ref{thm:bound1} can be achieved.
\begin{theorem}
\[ \liminf_{M\rightarrow\infty} \frac {L^e_{\min}(M)}{2M^2} = 1.
\]
\label{thm:tight}
\end{theorem}
\begin{proof}
Consider $M\geq 4$ is not a prime number.
By the CRTm construction, we can obtain equi-difference UI sequences of constant-weight $M+1$ and period $p(2M-1)$.
Since there are infinitely many primes $p$ and we can always set $M=p-1$, we have $\liminf_{M\rightarrow\infty} p/M=1$.
Therefore, we have
\[
\liminf_{M\rightarrow\infty} \frac {p(2M-1)}{2M^2} = 1.
\]
This shows that the asymptotic lower bound in Theorem~\ref{thm:bound1} is tight and hence the result follows.
\qed
\end{proof}

\section{Performance Study}\label{sec:simulation}

We consider that only a subset of $M$ users is active in a communication session due to different event-driven traffics. 
To facilitate our study, we assume that each user independently becomes active at the beginning of a communication session with an activation probability $p_a>0$.
It should be noted that, as the basic case when $p_a=1$ \cite{VANET13,SCCW16} is the design goal for deterministic protocol sequence based scheme to achieve UI property, we will investigate some more scenarios when $p_a<1$ by numerical study.

To our best knowledge, for all non-prime $M\geq 8$, CRTm and CRT sequences both have the shortest known sequence period for the UI property, but the Hamming weight of CRTm sequences is larger by one.
To show that CRTm sequences are of more practical interests, we present a performance comparison between CRTm sequences and CRT sequences through numerical study.
We also take the random access scheme into consideration to enrich the comparison.

\subsection{Activation probability $p_a=1$}

We first consider the case that all users are active.
As, for all non-prime $M\geq 8$, CRTm and CRT sequences both have the same sequence period, where the sequence period indicates the maximum delay guarantee, we consider two performance metrics: the \emph{average individual-delay} and \emph{average group-delay}.
Starting from a random time instant, the individual-delay of a user measures how long an active user has to wait until he or she can send one packet successfully, while the group-delay is the time duration we should wait until every active user has sent successfully at least one packet.
For each $M$, we run 500,000 samples to generate uniformly distributed random delay offsets for both CRT and CRTm sequences.

In the random access scheme, each active user sends a packet in each time slot with an independent probability $p_s$.
For a fair comparison with the same energy cost as CRTm sequences, we first pick $p_s=(M+1)/L$, where $L=p_M(2M-1)$ is the corresponding period of CRTm sequences.
Second, we find out the optimal $p_s$ that minimizes the expected delay.
Note that, for a given user $j$, the expected delay is given~\cite{VANET13} in the form 
\begin{equation}\label{eq:randomdelay}
\sum_{i=0}^\infty\big(1-1(1-\beta^i)^{M-1}\big),
\end{equation} 
where $\beta=1-p_s(1-p_s)^{M-1}$ is the probability that user $j$ either keeps silent or transmits a packet with collision.
One can see $p_s=1/M$ would minimize the expected delay in \eqref{eq:randomdelay}.
For each $M$, we generate 500,000 runs for both $p_s=(M+1)/L$ and $p_s=1/M$.

Table~\ref{table:delay} shows that the protocol sequence based scheme enjoys a better delay performance than the random access scheme for all examined cases.
In particular, even if the random access adopts the optimal transmission probability, the CRTm construction still enjoys a $9.8\%-20.5\%$ ($16.4\%-28.6\%$, resp.) improvement in average individual-delay (group-delay).
Furthermore, when $ M\leq 20$, the CRTm sequences have a $3.7\%-9.0\%$ ($4.1\%-10.4\%$, resp.) improvement in average individual-delay (group-delay) over CRT sequences, which can be attributed to the fact that the CRTm construction produces one more transmitting opportunity for each active user in every sequence period.
Such a slight improvement is still desirable in alarming applications, such as health care monitoring in body area networks, natural disaster prevention in wireless sensor networks and collision warning in vehicular networks, and can avoid more potential accidents.


\begin{table}[h]
\scriptsize
\caption{Delay performance of random access schemes, CRT sequences and CRTm sequences.}
\label{table:delay}
\begin{tabular}{llrrrrrr} \hline
 & $M$ & 8 & 9 & 10 & 12 & 14 &  15 \\ \hline \hline
 & Random $p_s=\frac{M+1}{L}$ & 27.1 & 29.0 & 30.9 & 37.5 & 47.1 & 48.9 \\
Average & Random $p_s=\frac{1}{M}$ & 20.4 & 23.1 & 25.8 & 31.3 & 36.7 & 39.4 \\
individual & CRT &20.1 &21.0 &21.9 &26.5 &34.5 &35.1 \\
-delay & CRTm &18.3 &19.5 &20.5 &25.1 &32.6 &33.5 \\ \hline
 & Random $p_s=\frac{M+1}{L}$ & 73.7 & 82.3 & 90.6 & 116.6 & 153.1 & 162.1 \\
Average & Random $p_s=\frac{1}{M}$ & 55.3 & 65.3 & 75.6 & 97.1 & 119.3 & 130.7 \\
group & CRT &51.1 &54.2 &58.2 &74.8 &104.1 &106.6 \\
-delay & CRTm &45.8 &49.9 &54.0 &70.5 &97.2 &101.3 \\
\hline \\ \hline
 & $M$ & 16 &18 &20 &25 &30  &40\\ \hline \hline
 & Random $p_s=\frac{M+1}{L}$ & 50.6 & 57.3 & 66.9 & 85.1 & 96.9 & 129.9 \\
Average & Random $p_s=\frac{1}{M}$ & 42.1 & 47.6 & 53.1 & 66.6 & 80.2 & 107.4 \\
individual & CRT &35.9 &40.5 &48.5 &62.0 &68.5  &91.7 \\
-delay & CRTm &34.4 &39.0 &46.6 &60.1 &66.9  &90.0 \\ \hline
 & Random $p_s=\frac{M+1}{L}$ & 171.3 & 200.5 & 241.1 & 324.8 & 387.0 & 555.7 \\
Average & Random $p_s=\frac{1}{M}$ & 142.4 & 166.2 & 191.0 & 254.0 & 320.7 & 458.8 \\
group & CRT &110.8 &129.4 &161.5 &221.2 &251.8  &361.6 \\
-delay & CRTm &105.7 &124.1 &153.9 &212.4 &245.1  &354.7 \\
\hline
\end{tabular}
\end{table}

\subsection{Activation probability $0<p_a<1$}

We broaden our discussion by examining the cases that $M=10,30$, and $p_a=90\%,80\%,\ldots,40\%$.
For each $M$ and $p_a$, each user has a probability $p_a$ to be active and a probability $1-p_a$ to remain idle in a communication session.
We run 500,000 samples to generate uniformly distributed random delay offsets for CRT and CRTm sequences.

It should be pointed out that for each given $M$, each user follows the same protocol sequences no matter what $p_a$ we consider, since that we need to guarantee the UI property in the worst case, i.e., all $M$ users are active.
Whereas, in the random access scheme, we consider different optimal $p_s$ that minimize
the expected delay for different $M$ and $p_a$.
As there are $Mp_a$ active users in average, by plugging $M=Mp_a$ into \eqref{eq:randomdelay}, one can see that the minimum expected delay for random access scheme occurs when $p_s=1/(Mp_a)$.
We generate 500,000 runs for each examined case.

Table~\ref{table:delay_partial_10} and Table~\ref{table:delay_partial_30} show that, for a fixed $M$, the average delay falls as $p_a$ decreases, which is due to the decrease of blocking probability of each submitted packet.
Since CRTm sequences have a better average delay performance than CRT sequences for all examined cases, in what follows we consider the comparison between CRTm sequences and the optimal random access scheme.

\begin{table}[h]
\scriptsize
\caption{Average delay performance of optimal random access scheme, CRT and CRTm sequences for $M=10$.}
\label{table:delay_partial_10}
\begin{tabular}{llcrrrrrr} \hline
 & $p_a$ & 100\% & 90\% & 80\% & 70\% & 60\% & 50\% & 40\% \\ \hline \hline
Average & Random $p_s=\frac{1}{p_a\cdot M}$ & 25.8 & 23.5 & 21.2 & 18.9 & 16.5 & 14.4 & 11.7 \\ 
individual & CRT & 22.9 & 21.6 & 20.4 & 19.2 & 18.1 & 17.3 & 15.8 \\
-delay & CRTm & 21.5 & 20.2 & 18.9 & 17.7 & 16.4 & 15.4 & 13.3 \\ 
\hline \\ \hline
 & $p_a$ & 100\% & 90\% & 80\% & 70\% & 60\% & 50\% & 40\% \\ \hline \hline
Average & Random $p_s=\frac{1}{p_a\cdot M}$ & 75.6 & 66.0 & 56.4 & 47.2 & 38.2 & 29.9 & 20.3 \\ 
group & CRT & 59.0 & 53.5 & 47.8 & 42.2 & 37.1 & 32.2 & 26.1 \\
-delay & CRTm & 55.1 & 49.1 & 43.4 & 38.2 & 33.4 & 28.3 & 21.9 \\  \hline
\end{tabular}
\end{table}

\begin{table}[h]
\scriptsize
\caption{Average delay performance of optimal random access scheme, CRT and CRTm sequences for $M=30$.}
\label{table:delay_partial_30}
\begin{tabular}{llcrrrrrr} \hline
 & $p_a$ & 100\% & 90\% & 80\% & 70\% & 60\% & 50\% & 40\% \\ \hline \hline
Average & Random $p_s=\frac{1}{p_a\cdot M}$ & 80.2 & 72.4 & 64.8 & 57.1 & 49.2 & 41.5 & 33.8 \\ 
individual & CRT & 69.5 & 65.0 & 60.8 & 56.7 & 52.9 & 49.3 & 45.7 \\
-delay & CRTm & 67.9 & 63.4 & 59.1 & 55.0 & 51.2 & 47.5 & 43.9 \\ 
\hline \\ \hline
 & $p_a$ & 100\% & 90\% & 80\% & 70\% & 60\% & 50\% & 40\% \\ \hline \hline
Average & Random $p_s=\frac{1}{p_a\cdot M}$ & 320.6 & 281.0 & 243.1 & 205.5 & 168.6 & 133.8 & 100.4 \\ 
group & CRT & 252.7 & 227.5 & 204.0 & 181.0 & 158.2 & 136.9 & 115.4 \\
-delay & CRTm & 245.9 & 221.1 & 197.2 & 174.0 & 151.2 & 129.4 & 108.3 \\  \hline
\end{tabular}
\end{table}

As shown in these two tables, the optimal random access scheme has a better average delay performance than the CRTm sequences when $p_a$ is small enough.
The reason is as follows.
For CRTm construction, given $M$, each user follows the same sequences no matter what $p_a$ we choose.
Therefore, as $p_a$ decreases, the active users access the channel more conservatively.
Such a behaviour causes more silent time slots and degrades the individual-delay performance of CRTm sequences.

Denote by the individual/group \emph{critical point} the activation probability $p_a$ when the random access scheme begins to have a better average/group individual-delay than the CRTm construction as $p_a$ decreases from $100\%$.
Simulation results show that for $M=10$ the individual (group, resp.) critical point may be between $50\%-60\%$ ($40\%-50\%$, resp.), and for $M=30$ the individual (group, resp.) critical point may be between $60\%-70\%$ ($40\%-50\%$, resp.).
There is a gap between individual and group critical points since the UI property of CRTm sequences makes the delay performance among all active users has a better fairness than the optimal random access scheme, as shown in Fig~\ref{fig:groupdelay}.
One can see from Fig~\ref{fig:groupdelay} that the maximum and minimum observed group-delays are closer for CRTm sequences.

\begin{figure}[h]
\centering
\includegraphics[width=3.5in]{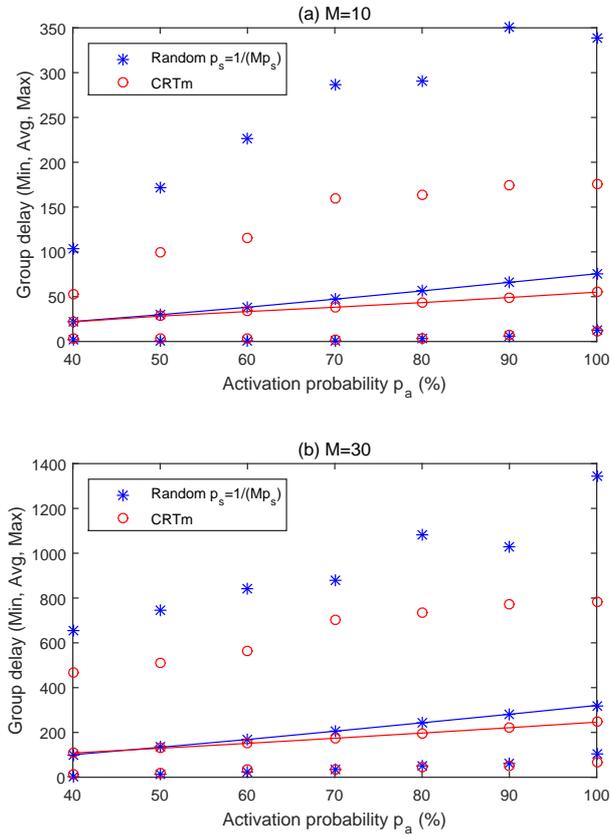}
\caption{The minimum, maximum observed group-delay and the average group-delay for (a) $M=10$ and (b) $M=30$. The average value is connected by a piece-wise linear curve. The symbols above and below this curve indicate the maximum and minimum value, respectively.} \label{fig:groupdelay}
\end{figure}

%
%

\begin{remark}
The CRTm sequences have the same influence on delay performance on the completely asynchronous channel, comparing against the CRT sequences and random access schemes.
\end{remark}


\section{Conclusion}\label{sec:conclusion}
This paper studies UI sequences for bounded-delay data service in wireless multiple-access networks without stringent time synchronization.
To achieve a better delay performance than that of previously known protocol sequences, a new construction of UI sequences with constant-weight $M+1$ is devised for a system shared by $M$ active users.
The new construction and the previously known CRT construction in \cite{SWSC} both produce the shortest known sequence period for all non-prime $M\geq 8$.
Since the two constructions have the same algorithm complexity, there is no difference in hardware cost when employing these two schemes.
Moreover, it is shown by numerical results that the new construction enjoys a better average delay performance than the optimal random access scheme and other constructions with the same sequence period, in a variety of traffic conditions, and thus is of more practical interests.
On the other hand, an asymptotic lower bound on the sequence period is derived for equi-difference UI sequences with constant-weight $M+1$, and is achieved by using the proposed new construction.

Our follow-up work seeks to investigate shorter UI sequences by using the methods presented in this paper.
In addition, our approaches can also be applied to study more general CACs, without requiring that the Hamming weight is equal to the number of active users.

\begin{acknowledgements}
The authors would like to express their gratitude to the anonymous referees for their valuable comments and suggestions.
This work was supported by the National Natural Science Foundation of China under grant number 61301107 and 11601454, the open research fund of National Mobile Communications Research Laboratory, Southeast University , under grant number 2017D09, and the Natural Science Foundation of Fujian Province of China under grant number 2016J05021.
\end{acknowledgements}


\end{document}